%% file: main.tex
\documentclass{article}

\usepackage[utf8]{inputenc} 
\usepackage[T1]{fontenc}
\usepackage{lmodern}
\usepackage{hyperref}       
\usepackage{url}            
\usepackage{booktabs}
\usepackage{siunitx}
\usepackage{amsfonts}       
\usepackage[numbers]{natbib}
\usepackage{nicefrac}       
\usepackage{microtype}      
\usepackage{algorithm,algpseudocode}
\usepackage{amsfonts}
\usepackage{amsmath}
\usepackage{amsthm}
\usepackage{amssymb}
\usepackage[title]{appendix}
\usepackage{bm}
\usepackage{courier}
\usepackage[usenames,dvipsnames]{color}
\usepackage{enumitem}
\usepackage{graphicx}
\usepackage[margin=0.9in]{geometry}
\usepackage{url}
\usepackage{subfigure}
\usepackage{multirow}
\usepackage{authblk}
\usepackage{subcaption}
\usepackage[dvipsnames]{xcolor}
\hypersetup{
	colorlinks,
	linkcolor={red!80!black},
	citecolor={blue!50!black},
	urlcolor={blue!80!black}
}

\newcommand{\ts}[1]{\textcolor{orange}{#1}}

\input{./Definitions}

\title{DuaLip-GPU Technical Report}

\author{Gregory Dexter, Aida Rahmattalabi$^{*}$, Sanjana Garg, Qinquan Song, Ruby Tu, Yuan Gao, Yi~Zhang, Zhipeng~Wang$^{*}$, Rahul Mazumder}

\affil{LinkedIn}

\date{} 

\usepackage{parskip}

\hypersetup{
	colorlinks,
	linkcolor={red!80!black},
	citecolor={blue!50!black},
	urlcolor={blue!80!black}
}

\begin{document}

\maketitle
\renewcommand{\thefootnote}{\fnsymbol{footnote}}
\setcounter{footnote}{0}

\footnotetext[1]{Corresponding authors: arahmattalabi@linkedin.com, zhipwang@linkedin.com}

\input{main_text}
\bibliography{bibliography}

\appendix

\input{appendix}

\end{document}

%% file: main_text.tex
\section{Introduction}
Linear programs (LPs) underpin many large-scale decision systems: ranking, allocation, and a broad class of matching problems solved on recurring cadences in production. Prior work in this lineage—\textsc{ECLIPSE}’s ridge-regularized dual ascent and LinkedIn’s open-source \textit{DuaLip} in a Scala/Spark stack—showed that first-order methods can scale to these settings, but the original system was tightly coupled to two fixed schemas and a CPU-centric runtime that made new formulations difficult to express and prevented effective use of modern accelerators.

We re-architect this solver for reuse and for modern hardware. The result is a Python library and system that co-designs (i) a compositional interface for plugging in problem logic, (ii) an algorithmic package suitable for extreme-scale LPs with decomposable ``simple constraints'' (defined subsequently), and (iii) GPU execution techniques needed to realize the available parallelism in practice. Rather than a single ``call a solver'' entry point, the library exposes operator-level roles that allow new LP formulations and projection maps representing simple constraint families to be added locally while the solve loop and diagnostics remain unchanged\footnote[2]{Code for the presented LP solver is available at \href{https://github.com/linkedin/DuaLip}{https://github.com/linkedin/DuaLip}.}.

The architecture is intentionally \emph{general}: it is designed for a broad class of LPs beyond any single schema. In this report, however, we focus on LPs with \emph{matching constraints} as our primary application domain, since these workloads dominate our production usage and expose the most stringent algorithmic and systems bottlenecks. We therefore use matching LPs as a running case study for the library and its implementation.

On these matching workloads, our design achieves order-of-magnitude wall-clock gains of \emph{at least $10\times$ faster} than the prior distributed CPU \textit{DuaLip} to reach a fixed sub-optimality gap while providing predictable scaling from one to many GPUs.

\textbf{Contributions.} This work brings together three strands informed by production use:
\begin{itemize}
  \item \emph{Operator-centric programming model and library interface.} We replace DuaLip’s schema-bound, template-based interface (multi-objective optimization vs.\ single-block matching) with an operator-centric model tailored to ridge-regularized dual ascent. Problem logic is expressed only through three primitives: \texttt{objectives/} (encapsulating the data and dual gradient computation), \texttt{projections/} (blockwise projections onto a decomposable constraint polytope), and an \texttt{optimizer/} (which performs dual ascent using the dual gradient), while the solve loop, diagnostics, and distributed primitives are shared. This separation enables support for new formulations—such as incorporating multiple constraint families—with minimal additional code and purely local composition.

  \item \emph{Algorithmic improvements for ridge-regularized dual ascent.} Building on the ECLIPSE/DuaLip framework, we strengthen the underlying method with Jacobi-style row normalization, primal scaling and a simple continuation scheme for the regularization parameter $\gamma$. We analyze the impact of these changes on conditioning and first-order optimizer convergence.

  \item \emph{GPU execution and sparse layout for extreme-scale LPs.} We show how to realize the available parallelism on GPUs by choosing sparse layouts aligned with the constraint structure, batching projection operators to amortize kernel launches, and using a distributed pattern that communicates only the dual variables. For the structured matching workloads we study, these system-level choices alone provide at least $10\times$ wall-clock speedup over the prior distributed CPU DuaLip under matched stopping criterion.
\end{itemize}

Together, these changes turn ridge-regularized dual ascent from a specialized, schema-bound deployment into a flexible, high-performance solver architecture suitable for modern GPU-accelerated pipelines. In the remainder of the report, we formalize the matching LPs that motivate our design, describe the operator-centric programming model, develop the algorithmic enhancements, and detail the GPU execution mechanisms, with experiments on extreme-scale matching workloads.

\section{Related Work}

Our work builds on multiple threads of prior research.

\textbf{First-order methods for large-scale linear programs.}  
Projected first-order methods have recently emerged as a practical approach for massive LPs, most notably PDLP~\cite{applegate2021practical} and its GPU implementation cuPDLP~\cite{lu2025cupdlp}. These solvers provide strong theoretical convergence guarantees and demonstrate competitive performance on high-precision benchmarks, including routing problems and general LP instances. However, their implementations rely on carefully engineered custom kernels optimized for generic sparse matrix–vector operations, thereby treating the constraint matrix as an unstructured operator. As a result, they do not explicitly leverage dual decomposition or diagonal-block structure and have not been demonstrated to scale to the extreme problem sizes characteristic of our target workloads.

In contrast, our setting targets extreme-scale matching and allocation workloads with highly structured constraint families and moderate precision requirements. Rather than pursuing uniform high-accuracy solutions, we emphasize rapid convergence to economically meaningful dual solutions under decomposable constraint structure. Our approach explicitly exploits diagonal-block and matching-specific structure in the dual updates, enabling distributed decomposition across columns and efficient aggregation via collective communication. This structural specialization differentiates our method from general-purpose first-order LP solvers.

\textbf{Large-scale matching and allocation.}  
A substantial literature addresses weighted matching and assignment at scale, including specialized bipartite matching solvers and network-flow-based formulations. While these approaches achieve strong performance for canonical matching constraints, they are typically restricted to fixed constraint classes (e.g., degree constraints, capacity constraints) and do not naturally accommodate additional arbitrary linear constraints that arise in production allocation systems. 

In contrast, our formulation preserves the flexibility of general linear programming while retaining the computational advantages of matching structure in the dual. This allows us to encode heterogeneous constraint families without modifying the core solver.

\textbf{Ridge-regularized dual ascent.}  
The ridge-regularized dual formulation introduced by \citet{basu2020eclipse} demonstrated that adding $\ell_2$ regularization to the primal induces smoothness in the dual, enabling scalable first-order optimization for extreme-scale allocation. \citet{ramanath2021efficient} further developed efficient projection algorithms and integrated these ideas into a reusable Scala/Spark system (DuaLip).

However, prior implementations were tightly coupled to CPU-centric data formats and rigid declarative schemas. In particular, the system exposed only a small number of fixed templates (e.g., single-block matching and multi-objective optimization), making it difficult to express problems with multiple interacting constraint families or more general block structure without modifying the solver internals. These limitations motivated a redesign centered on composability, flexible constraint specification, and GPU-native execution.

\textbf{High-performance tensor frameworks.}  
Our implementation leverages modern array-programming frameworks such as PyTorch~\cite{paszke2019pytorch}, enabling concise expression of sparse–dense operations, batched projections, and distributed collectives. Unlike custom low-level kernels tailored to a fixed LP abstraction, this approach allows rapid iteration over algorithmic variants and structural formulations while retaining high performance on contemporary GPU clusters.

\medskip
Together, these strands of work inform the programming model, algorithmic design, and distributed execution mechanisms underlying the present solver.

\section{Problem Setting}
Many large-scale optimization tasks in industry reduce to solving very large linear programs on a recurring cadence, e.g., daily, weekly, or per model update. Examples include email and notification allocation in ECLIPSE \cite{basu2020eclipse}, content and opportunity assignment in DuaLip \cite{ramanath2021efficient}, and recent marketplace-shaping problems such as promoting inactive members \cite{acharya2023promoting}. In each case, the LP instance changes as new scores, eligibility signals, or demand arrive, but the underlying structure of the formulation remains stable.

In such recurring regimes, the primary challenge is not merely solving a single large LP to high precision, but ensuring stability, scalability, and fast convergence as the data change across runs. Classical LP formulations are inherently non-strongly convex and often exhibit degeneracy. Introducing $\ell_2$ regularization in the primal induces strong convexity, yielding a smooth dual objective with Lipschitz-continuous gradients. This results in faster convergence under first-order dual methods, and creates a well-conditioned structure that is amenable to distributed decomposition. In recurring industrial workloads, these properties are often more critical than attaining an exact vertex solution of the original LP.

\subsection{Preliminaries}\label{sxn:prelim}
We are given an LP of the form:
\begin{gather}
    \min~~\cbb^T\xb~~\text{such that }\Ab\xb \leq \bb, \xb \in \Ccal, \label{eqn:LP}
\end{gather}
where $\xb \in \R^n$ is the vector of primal variables, $\cbb$ is the vector of primal objective coefficients, $\Ab \in \R^{m \times n}$ and $\bb$ specify the ``complex'' constraints, and $\Ccal$ denote the ``simple'' constraints polytope.

We adopt the algorithmic approach of ECLIPSE \cite{basu2020eclipse}, which perturbs the LP with ridge regularization and maximizes the resulting dual function using gradient-based methods. This approach is particularly efficient when the primal variables lie in certain common polytopes, which we describe later. Concretely, ECLIPSE augments the primal objective in Eq. (1) with a ridge term $\frac{\gamma}{2}\|\mathbf{x}\|_2^2$ and solves the resulting Lagrangian relaxation.

Let $\lambdab \in \R^m$ denote the dual values of the complex constraints. The dual objective is defined as:
\begin{gather}
    g(\lambdab) = \min_{\xb \in \Ccal} \cbb^T\xb + \frac{\gamma}{2}\|\xb\|_2^2 + \lambdab^T(\Ab\xb - \bb)
\end{gather}
By Danskin's theorem, $\nabla \gb(\lambdab)$ may be efficiently computed by $\nabla \gb(\lambdab) = \Ab\xb_\gamma^*(\lambdab) - \bb$, where $\Pi$ is the projection to $\Ccal$, and: 

\begin{gather*}
    \xb_\gamma^*(\lambdab) = \Pi_{\Ccal}\left(\frac{-1}{\gamma}(\Ab^T\lambdab + \cbb)\right).
\end{gather*}
Any optimization algorithm that depends only on the gradient $\nabla \gb(\lambdab)$ may then be used to maximize $g(\lambdab)$ over $\lambdab \geq \zero$, and hence the perturbed primal LP by strong duality. We take this algorithmic approach due to the substantial reduction in communication overhead when $\Pi$ can be applied in a distributed manner. This is the case for many important applications of extreme-scale linear programming in industry. 

\subsection{Matching Problem Formulation}


Despite the generality of our approach, we focus on extreme-scale matching problems in which users are assigned to items subject to shared impression-capacity constraints. Formally, let the decision variable $x_{ij}$ denote the assignment of a source (e.g., user) $i \in [I]$ to a destination (e.g., campaign or item) $j \in [J]$. For each user $i$, define the block vector
\[
\mathbf{x}_{i} := (x_{i1}, \ldots, x_{iJ})^\top,
\]
and stack these blocks to form the full decision vector
\[
\mathbf{x} = \big[\, \mathbf{x}_{1}; \ldots; \mathbf{x}_{I} \,\big].
\]
Let $E \subseteq [I] \times [J]$ denote the set of feasible source–destination pairs.

The coupling constraints arise from shared destination-level requirements—such as budget, pacing, frequency, or delivery limits that apply collectively to all users eligible for a given destination. For instance:

\begin{align}
\sum_{i \in [I]} a_{ij} x_{ij} &\le d_j 
&& \forall j \in [J], \quad \text{(destination capacity)}
\end{align}
In matrix form, $\Ab$ can be written as 
\[
\Ab =
\begin{bmatrix}
a_{11} & \cdots & 0 & a_{21} & \cdots & 0 & a_{I1} & \cdots & 0 \\
\vdots & \ddots & \vdots & \vdots & \ddots & \vdots & \vdots & \ddots & \vdots \\
0 & \cdots & a_{1J} & 0 & \cdots & a_{2J} & 0 & \cdots & a_{IJ}
\end{bmatrix}, 
\]
which can be viewed as a horizontal concatenation of diagonal blocks. 
Consequently, $\Ab$ is highly sparse: entries are zero both due to the block-diagonal structure and because many coefficients $a_{ij}$ vanish for ineligible pairs. This structured and unstructured sparsity is a key property that we exploit to enable efficient matrix--vector multiplications, as discussed in Section~\ref{sec:gpu_enhancement}.

In addition, we often impose user-level impression capacity constraints that can be expressed as simple ``box'' or ``simplex'' constraints~\cite{ramanath2021efficient}, applied independently to each block $\mathbf{x}_{i} \in \Ccal_i$. In extended form, these constraints are written as:
\begin{align}
\sum_{j \in [J]} x_{ij} &\le 1 
&& \forall i \in [I], \quad \text{(per-user capacity)} \\
0 \le x_{ij} 
&&& \forall (i,j) \in E. \quad \text{(feasibility)}
\end{align}

\paragraph{Simple constraints}
In the above example, constraints (4), (5) form simple polytopes (simplex) that admit efficient projection algorithms. More generally, when assigning user to items the impression capacity of a user may be represented through a ``box-cut'' or ``simplex'' inequality constraint, which can be applied to each block $\xb_{i*}$ independently \cite{ramanath2021efficient}. We refer to these constraints as \emph{simple constraints}. See~\cite{ramanath2021efficient} for detailed examples of such polytopes and their associated projection procedures.

Exploiting this separable structure allows the user-side constraints to be handled implicitly within the primal subproblem, without introducing explicit dual variables for them. Consequently, when matching $I$ users to $J$ items, the number of explicit dual variables can be reduced from $I + J$ to $J$. This substantially reduces communication overhead, since only the destination-level dual variables must be synchronized across GPUs when computing $\nabla \gb(\boldsymbol{\lambda})$.

\paragraph{Matching constraints}
We now describe the structure of the complex constraints $\Ab\xb\le \bb$ for matching problems. This structure closely follows the formulation in \textsc{ECLIPSE}~\cite{basu2020eclipse}, which motivated the modern ridge-regularized dual-ascent approach, but is more general than the version implemented in Scala/Spark DuaLip, which permitted at most a single matching block. Here we allow an arbitrary number of matching constraint families, each interacting with user blocks in a simple element-wise fashion.

In matching problems, the complex constraints couple many users through shared destinations—budget, pacing, frequency, or delivery constraints that apply to all $i$ such that $(i,j)\in E$. Because each such constraint enters the LP in an element-wise manner on each $\xb_{(i)}$, the constraint matrix has a diagonal block structure across sources.

\vspace{0.5cm}
\begin{definition}[Complex constraint matrix for matching problems]
\label{def:matching-constraints}
Given decision variables $\xb$ stacked as above, the complex constraint matrix $\Ab \in \R^{mJ\times IJ}$ is defined blockwise across sources as
\begin{equation}
\label{eq:A-matching}
\Ab \;=\;
\begin{bmatrix}
\Db_{11} & \Db_{12} & \cdots & \Db_{1I}\\
\Db_{21} & \Db_{22} & \cdots & \Db_{2I}\\
\vdots   & \vdots   & \ddots & \vdots  \\
\Db_{m 1} & \Db_{m 2} & \cdots & \Db_{m I}
\end{bmatrix},
\qquad
\Db_{k i} \in \R^{J\times J}
\text{ diagonal for every }(k,i),
\end{equation}
where $k=1,\dots,m$ indexes distinct constraint families. Each diagonal block $\Db_{k i}$ acts element-wise on the $i$th variable block, $\xb_{(i)}$.
\end{definition}

This formulation captures the structured matching constraints used in \textsc{ECLIPSE} while allowing multiple such families (budget, pacing, fairness, frequency, etc.) to appear simultaneously. The diagonal structure of $\Db_{ki}$ is a key property exploited throughout the paper: it simplifies the dual gradient, enables efficient sparse layouts and projection batching on GPUs, and allows for sharp conditioning results such as those in Section~\ref{sxn:algorithm_enhancments}.

\section{Programming Model and Library Architecture}
\label{sxn:programming-model}

The previous Scala Spark version of DuaLip \cite{dualipcode} enabled the first production deployments, but it supported only two rigid schemas and relied on an object model that made extensions costly. Its runtime also fit poorly with common numerical tooling and Python workflows, so diagnostics and tuning were largely manual. We re‑architect LP solving in an \emph{imperative, operator‑level array/tensor programming model} (PyTorch‑style, define‑by‑run) rather than a task‑level “call a solver” API. Concretely, the hot path is an explicit dataflow of sparse matrix–vector operators \((\Ab \xb,\ \Ab^\top \lambdab)\) and blockwise projections, orchestrated by a small maximizer. This boundary is deliberate: it exposes the kernels that dominate runtime, lets us vary sparse layouts/kernels and projection operators, and maps cleanly to GPU execution, without rewriting the solve loop.

The library is built to support this \emph{general} setting: LPs with (i) simple constraints that decompose cleanly over blocks of variables, making projection-based methods natural, and (ii) sparse complex constraints for which the dominant operations are repeated applications of $\Ab \xb$ and $\Ab^\top \lambdab$. This covers a broad class of assignment, allocation, and routing problems encountered by vertical teams across the organization, and is the level at which the programming model and solver interface are designed.

Within this general template, \emph{matching} problems are a particularly important and pervasive subclass. Here, each “source’’ has its own block of variables and simple constraints, while cross-source constraints, e.g., budgets, pacing, and frequency caps, interact with each block through simple element-wise terms. As shown in prior work, this yields a distinctive diagonal block structure in the complex constraint matrix. We formalize this structure in Definition~\ref{def:matching-constraints} and show that it enables highly efficient GPU kernels, batched projection operators, and effective preconditioning. While the focus of this paper is on matching workloads, the library itself is not restricted to a single matching block or even to matching constraints alone; the same abstractions apply to any LP whose simple constraints decompose over blocks and whose complex constraints can be implemented via sparse operator calls.

While the remainder of the paper uses matching LPs as a running example, the library interfaces in this section make no reference to matching-specific notions. The library consists of three decoupled components with minimal contracts: 

\begin{table}[h]
\centering
\renewcommand{\arraystretch}{1.15}
\setlength{\tabcolsep}{6pt}
\begin{tabular}{@{}p{3.0cm}p{5.2cm}p{7.2cm}@{}}
\hline
\textbf{Component} & \textbf{Constructed from} & \textbf{Must provide (single method)} \\
\hline
\textbf{Maximizer} &
Algorithm settings (step policy, iteration budget, flags) &
\texttt{maximize(obj, initial\_value)} $\rightarrow$ \texttt{Result} \\
\textbf{ObjectiveFunction} &
LP tensors and metadata ($\Ab,\bb,\cbb$) plus a supplied \texttt{ProjectionMap} &
\texttt{calculate($\lambda$, $\gamma$)} $\rightarrow$ \texttt{ObjectiveResult} \\
\textbf{ProjectionMap} &
Mapping from primal columns/blocks to projection operators (simplex, box-cut, \dots) &
\texttt{project(block\_id, v)} $\rightarrow$ projected $v$ \\
\hline
\end{tabular}
\caption{Core objects, what constructs them, and the single required method each exposes.}
\label{tab:arch-minimal}
\end{table}

A new LP formulation implements an \texttt{ObjectiveFunction}, which encapsulates the LP parameters and a method for computing $\nabla \gb(\lambdab)$ from a dual value $\lambdab$. In practice, this involves ensuring that the matrix-vector products $\Ab\xb$ and $\Ab^T\lambdab$ can be efficiently represented and computed on GPU. Additionally, the computation of the gradient requires computing the projection $\Pi(\xb_\gamma^*(\lambdab))$ to the simple constraint polytope, $\Ccal$. The simple constraints tend to be broadly applicable, e.g., per-user simplex constraints or unit box constraints, and so the projection operators may be reused across new LP formulations. The \texttt{ObjectiveFunction} class provides an encapsulation of the LP formulation that exposes a method to compute the gradient. This means that \texttt{Maximizer} may be a variety of optimization algorithms that perform the dual ascent of $\gb(\lambdab)$ over the dual feasible space $\lambdab \geq \zero$.

We take the stance that solving extreme-scale LPs necessitates a level of efficiency that requires flexible low-level modification of the core algorithmic primitives, particularly the gradient computation, in order to be applied across various formulations in practice. At the same time, it is necessary that these modifications take minimal engineering effort, and that shared primitives and sensible defaults can be shared across applications. This solver architecture has the advantage that the major algorithmic components are reflected in the code structure, and so modifying any part only require local modifications. The total solver for a use case is then a composition of the high-level components, much like how a model in PyTorch is constructed via  high-level composition of the model architecture, loss function, optimization algorithm, etc.

In contrast, the Spark Scala-based solver provided a single declarative entry point with support for the two strict schemas that then invoke the necessary algorithm components through inheritance hierarchies. We have found that this structure makes it difficult to accommodate new constraints even when they are efficient to represent algorithmically, such as a global count constraint in a matching problem, $\sum_{ij} x_{ij} \leq m$. While it's trivial to compute $\Ab\xb$ and $\Ab^\top \lambdab$ for this constraint, appending it to the matching problem in the Spark Scala solver requires extensive changes across the code base.

The implementation is Python‑native to align with common numerical stacks and to support instrumentation and iteration. Optional hooks for experiment logging/tuning, structured diagnostics, and standard debugging wrap the same calls without altering these interfaces.


\section{Optimization Algorithm}\label{sxn:algorithm_enhancments}

Our solver follows the ridge-regularized dual ascent framework of Section~\ref{sxn:prelim}: given the smoothed dual $g(\lambdab)$, we run a first-order method that accesses $g$ through its gradient $\nabla \gb(\lambdab) = \Ab\xb_\gamma^*(\lambdab) - \bb$ and the projection operators that define $x_\gamma^*(\lambdab)$. The Scala/Spark DuaLip implementation instantiated this framework with a variant of Nesterov Accelerated Gradient Descent (AGD) and a small set of alternative optimizers, which required careful, instance-specific tuning to obtain reliable convergence on large matching problems and was sensitive to the choice of scaling and regularization.

In this section, we keep the same dual formulation but redesign the optimization stack so that a \emph{single} configuration works robustly across the matching workloads of Definition~\ref{def:matching-constraints}. Concretely, we introduce three changes: (i) conditioning of the dual problem via Jacobi-style row normalization, (ii) a simple continuation schedule for the ridge parameter $\gamma$ that balances curvature against approximation error, and (iii) balancing the smoothing of the ridge term along coordinates via primal coordinate scaling.  Conceptually, these techniques apply whenever we have access to the constraint matrix $\Ab$ (and its row/column statistics), the simple-constraint projections, and the smoothed dual objective, so they are compatible with the operator-centric programming model of Section~\ref{sxn:programming-model}. 

\subsection{Algorithm enhancements}

A central area of focus in improving the DuaLip solver was to improve the performance and robustness of the ridge-regularized dual ascent method that underpins this line of work. The basic framework of Section~\ref{sxn:prelim} admits many first-order instantiations; in practice, the dominant convergence issues we observed on production matching workloads stemmed from three sources: (i) poor conditioning of $\Ab \Ab^\top$ in the dual, (ii) tradeoff between convergence speed and solution quality with the ridge regularization term and (iii) heterogeneous scales in the primal variables that interact badly with the quadratic regularizer. We address these in turn via Jacobi-style row normalization, a simple regularization schedule, and primal scaling. 


\paragraph{Jacobi preconditioning / row normalization.}
Our first enhancement is to improve the conditioning of the dual problem by rescaling the complex constraints. We assume $\Ab$ is full row rank. Intuitively, when some rows of $\Ab$ have much larger norms than others, gradient steps for the smoothed dual move too cautiously in some directions and too aggressively in others. 

When the projection operator is inactive (or the identity), the dual gradient reduces to
\[
\nabla \gb(\lambdab)\;=\;-\frac{1}{\gamma}\bigl(\Ab\Ab^\top\lambdab+\Ab\cbb\bigr)\;-\;\bb,
\]
so the Hessian is
\[
\nabla^2 \gb(\lambdab)\;=\;-\frac{1}{\gamma}\,\Ab\Ab^\top.
\]
The convergence of first-order methods on $g$ is therefore governed by the conditioning of $\Ab\Ab^\top$.

We apply a standard row-scaling transform. Let
\[
\Db=\mathrm{diag}\!\bigl(\|\Ab_{1*}\|_2^{-1},\ldots,\|\Ab_{m*}\|_2^{-1}\bigr)
\]
for all rows with nonzero norm (rows with $\|\Ab_{r*}\|_2=0$ are redundant and may be dropped or left unscaled with $\Db_{rr}=1$), and define the row-scaled system
\[
\Ab'=\Db\Ab,\qquad \bb'=\Db\bb.
\]
Because $\Db$ has positive diagonal entries, row scaling preserves the feasible set exactly:
\[
\{\,\xb:\Ab\xb\le\bb\,\}\;=\;\{\,\xb:\Ab'\xb\le\bb'\,\}.
\]
Moreover,
\[
\Ab'\Ab'^\top\;=\;\Db(\Ab\Ab^\top)\Db,
\qquad
\Db^2=\mathrm{diag}(\Ab\Ab^\top)^{-1}\ \ \text{(on nonzero rows)},
\]
so row normalization is precisely Jacobi preconditioning of the dual Hessian $-\nabla^2 \gb(\lambdab)=\frac{1}{\gamma}\Ab\Ab^\top$.

For the matching constraint matrix of Definition~\ref{def:matching-constraints}, $\Ab$ is a horizontal concatenation of diagonal subblocks across sources, so $\Ab\Ab^\top$ is a sum of (nearly) diagonal matrices (one per source) and is therefore close to diagonal in practice. Enforcing $\mathrm{diag}(\Ab'\Ab'^\top)=\Ib$ tightly clusters the spectrum; in the ideal diagonal/orthogonal case it yields $\Ab'\Ab'^\top=\Ib$ and condition number~$1$.

The following lemma formalizes this intuition under a simple statistical model for the matching blocks.

\vspace{0.5cm}
\begin{lemma}
\label{lemma:precon}
Let $\Ab=[\Ab_1~\cdots~\Ab_I]\in\R^{mJ\times IJ}$ with user blocks
$\Ab_i\in\R^{mJ\times J}$ that are i.i.d.\ across $i$ and diagonal by rows as in
Definition~\ref{def:matching-constraints}. Let $\Db_{\mathrm{exp}}=\mathrm{diag}\!\big(\EE\|\Ab_{1*}\|_2^2,\ldots,\EE\|\Ab_{m_2*}\|_2^2\big)^{-1/2}$
and $\widetilde{\Ab}=\Db_{\mathrm{exp}}\Ab$. Then $\mathrm{diag}\!\big(\EE[\widetilde{\Ab}\widetilde{\Ab}^\top]\big)=\Ib$.
If, in addition, for $r\neq s$,
\[
\EE\!\big[\langle \Ab_{r*},\Ab_{s*}\rangle\big]
\;\le\; \eta\,\sqrt{\EE\|\Ab_{r*}\|_2^2\,\EE\|\Ab_{s*}\|_2^2}
\quad\text{with }\eta\in[0,1),
\]
then
\[
\kappa\!\big(\EE[\widetilde{\Ab}\widetilde{\Ab}^\top]\big)
\;\le\; \frac{1+(m-1)\eta}{\,1-(m-1)\eta\,}.
\]
\end{lemma}

Thus, under mild cross-row correlation assumptions, row normalization nearly equalizes the eigenvalues of $\Ab\Ab^\top$ in expectation for matching workloads, which stabilizes dual first-order updates.

\paragraph{Regularization decay} The preconditioning above addresses the dominant factor when the simple constraints are inactive: in that regime, the condition number of $\Ab\Ab^\top$ controls the convergence rate of first-order methods. However, the ridge parameter $\gamma$ enters separately through the \emph{smoothness} (and effective strong convexity) of the smoothed dual. 

Next, we consider how to set the regularization hyperparameter over the course of optimization. While $\gamma$ does not affect the conditioning of $\Ab\Ab^\top$, it has a substantial effect on the smoothness of the objective. On the one hand, one prefers small $\gamma$ to avoid perturbing the original LP. Lemma 2 in \citet{basu2020eclipse} guarantees that there is always a range of sufficiently small $\gamma$ values that allow exact recovery of a solution to the un-smoothed LP. On the other hand, the Lipschitz constant of the gradient is only bounded by $\|\Ab\|_2^2/\gamma$ (see Lemma 3 in \cite{basu2020eclipse}), so very small $\gamma$ leads to a poorly conditioned dual problem and slow progress in practice.

To balance convergence speed and solution fidelity, we implement a simple continuation scheme in which $\gamma$ is initialized at a moderately large value (for stable, fast early progress) and then decayed on a pre-specified schedule and rate as the algorithm approaches a good dual solution. Since $\gamma$ directly affects the smoothness of the dual objective, we scale the maximum AGD step size proportionally with the decay of $\gamma$ to maintain stability across transition points. Overall, this approach allows for faster convergence in early iterations, while ensuring the final solution is not greatly perturbed from the true LP solution.

\paragraph{Primal scaling}
A limitation of the ridge-regularized formulation used in \textsc{ECLIPSE} and \textsc{DuaLip} is that the quadratic smoothing term $\frac{\gamma}{2}\|\xb\|_2^2$ implicitly assumes all coordinates of $\xb$ are on a comparable scale.  
When some entries of $\xb$ differ in magnitude by several orders, the regularization either dominates small coordinates or becomes negligible for large ones.  
For moderate $\gamma$, the term $\frac{\gamma}{2}\|\xb\|_2^2$ can distort the objective geometry and shift the optimal solution away from the true LP optimum, while for extremely small $\gamma$ the smoothing becomes ineffective and the dual landscape loses the curvature needed for fast convergence.

A simple remedy is to introduce a diagonal scaling of the primal coordinates.  
Let $\vb\in\R^n_{>0}$ be a vector of positive scale factors, and define the scaled variables:
\[
\zb \;=\; \Db_v\,\xb, \qquad \Db_v := \mathrm{diag}(\vb),
\]
so that $\xb=\Db_v^{-1}\zb$.  
Substituting into the regularized primal 
yields
\[
\min_{\zb \in \Db_\vb^{-1}\Ccal}\; (\Db_v^{-1}\cbb)^\top\zb \;+\; 
\tfrac{\gamma}{2}\|\zb\|_2^2
\quad\text{s.t.}\quad \Ab\Db_v^{-1}\zb\le\bb, 
\]
where \[
\Db_{\vb}^{-1}\Ccal
\;:=\;
\{\zb \in \mathbb{R}^n \;:\; \Db_{\vb}^{-1}\zb \in \Ccal\}.
\]

Equivalently, 
\[
\min_{\xb\in\Ccal}\; \cbb^\top\xb \;+\; \tfrac{\gamma}{2}\|\Db_v\xb\|_2^2
\quad\text{s.t.}\quad \Ab\xb\le\bb.
\]
Hence, two algebraically equivalent perspectives arise:
(i) generalize the regularizer to $\tfrac{\gamma}{2}\|\xb\|_{\Db_v^{2}}^2 = \frac{\gamma}{2}\xb^\top\Db_v^{2}\xb$, i.e.\ replace $\gamma\Ib$ by $\gamma\Db_v^{2}$, or
(ii) rescale the primal space by $\Db_v$, setting
\[
\cbb' = \Db_v^{-1}\cbb,\qquad 
\Ab' = \Ab\Db_v^{-1}.
\]
Either view leads to the same optimization problem after change of variables.
Choosing $\vb$ according to typical magnitudes of the primal coordinates
(or the column norms of $\Ab$)
balances the contribution of each coordinate in the ridge term and normalizes the effective curvature of the smoothed dual.
This scaling prevents the regularizer from either overwhelming or vanishing on particular coordinates, yielding a well-conditioned dual objective for any moderate~$\gamma$.




\section{GPU System/Algorithm Enhancements}\label{sec:gpu_enhancement}
We now turn to the system design needed to run this solver efficiently on modern hardware. In production settings, new formulations must converge reliably with minimal engineering effort and without extensive hyperparameter retuning. The ridge-regularized dual ascent algorithm we adopt from prior work is particularly well suited to this regime: its smoothing improves stability across recurring solves, and each iteration parallelizes cleanly across blocks while requiring synchronization only of the dual variables. This makes the method naturally compatible with distributed GPU execution and low-overhead deployment of new, stable formulations.

The matching LPs of Definition~\ref{def:matching-constraints} expose massive sparse linear operators and blockwise projections that map well to GPUs, but the original Scala/Spark DuaLip stack is tied to a JVM/CPU runtime and cannot easily leverage accelerator-friendly layouts or communication patterns. Our goal in this section is to capture most of the performance benefits that custom C++/CUDA implementations (e.g., PDLP-style systems) would offer, while keeping the implementation lightweight by building on PyTorch’s sparse and distributed primitives. We do so by choosing sparse tensor layouts that respect the matching constraint structure, batching projections to avoid tiny kernels, and using a simple multi-GPU execution pattern that reuses the same code across formulations.

\paragraph{Sparse tensor computation} 
In practice, the constraint matrix in Definition \ref{def:matching-constraints} has two sources of sparsity (i) structured sparsity from the block diagonal form, and, (ii) unstructured sparsity within the diagonal blocks since only a fraction of source-destinations pairs are eligible due to external factors. 

We may efficiently operate on such constraint matrices as follows. Let $\Ab$ as in Definition \ref{def:matching-constraints} with a single matching constraint family, i.e., 
\begin{gather*}
    \Ab = [\Db_1, \hdots, \Db_I], \quad \xb = [\xb_{(1)};\hdots; \xb_{(I)}]
\end{gather*}


We exploit both the structured and unstructured sparsity by storing constraint matrices in Compressed Sparse Column (CSC) format with columns ordered by destination \(j\), so all variables corresponding to a given source live contiguously in memory. Each column of the sparse tensor representation,  $\Tcal$, is such that $\Tcal[:][i] = \diag(\Db_i)$.
 
At fixed numeric precision, CSC stores exactly the structure we need: one integer pointer per column, one integer row index per nonzero, and one value per nonzero. Additionally, this allows us to immediately invoke the entry-wise arithmetic operations for the dual ascent algorithm with optimized GPU tensor operations. In contrast, the tuple/sequence approach used in the Scala version incurs pointer/boxing overhead, poorer cache locality, and conversion/sorting before math kernels, all of which raise memory traffic and wall-time without adding information.

\paragraph{Batched projection operator} Columns in the CSC layout tensor, $\Tcal$, align with the destination-partition of \citet{basu2020eclipse}, which is ideal for sparse matrix-vector products when leveraging the block diagonal structure, but it is still inefficient for per-slice projections if launched one column at a time on GPU (tiny kernels, launch overhead, low occupancy). We therefore keep the CSC tensor representation but execute projections on \emph{dense, padded batches}. 

The batching scheme uses logarithmically spaced buckets: we group columns by slice length into ranges \([2^{t-1}, 2^t)\). For each bucket we gather the relevant slices into a single dense slab (padding to the bucket’s upper bound), apply one batched projection kernel (e.g., unit-box clamp or simplex), and scatter the results back. Geometric bucketing limits padding waste to below a factor of two within each bucket, so temporary memory remains within a small constant factor of the true nonzeros, while the number of GPU launches is only \(1+\lfloor \log_2 s_{\max}\rfloor\), where \(s_{\max}\) is the maximum slice length.

This turns many irregular, low-occupancy calls into a handful of high-occupancy, coalesced operations without abandoning the compact CSC storage. It also integrates cleanly with PyTorch’s batched dense kernels, while the prior sequence-of-tuples style hampers efficient packing and requires extra conversions before math.

\paragraph{Distributed GPU communication}
  We use \texttt{torch.distributed} with the \texttt{nccl} backend and a single
  process group across all GPUs.  Columns of $\Tcal$ (and $\cbb$, consistently)
  are partitioned across devices in a balanced column split of the CSC-format
  matrices, while $\lambdab$ and the constraint vector $\bb$ are replicated on
  every device.  Rank~0 generates and partitions the full dataset on CPU, then distributes each
  rank's column partition via a collective \texttt{scatter}, followed by a
  \texttt{broadcast} of the shared constraint vector~$\bb$.  For a
  fixed $\lambdab$, each device computes its local slice of
  $\xb_\gamma^*(\lambdab)$ and its contribution to the dual quantities,
  leveraging the simple-constraint decomposition (per-$i$ projections and
  destination-separable blocks) to avoid cross-device dependencies during
  computation.

The communication volume per step depends only on the dual dimension $|\lambdab|$, not on the number of nonzeros or the per-GPU column partition. Each iteration proceeds as follows:
\begin{enumerate}[nosep]
\item Every rank computes its local gradient contribution, local objective
      value, and local regularization penalty concurrently on its own CUDA
      device.
\item A \texttt{reduce} (SUM, to rank~0) aggregates the gradient
      ($|\lambdab|$ floats) and two scalars (objective, regularization).
\item Rank~0 performs the accelerated gradient-ascent update, producing new
      iterates $\lambdab_1$ and $\lambdab_2$.
\item Two \texttt{broadcast} calls (from rank~0) distribute updated
      $\lambdab_1$ and $\lambdab_2$ vectors, which form the momentum state used in the AGD update, to all other ranks. 
\end{enumerate}
The per-step communication therefore consists of one \texttt{reduce} and two
\texttt{broadcast} operations, each of size $|\lambdab|$, plus $O(1)$ scalar
reductions---independent of sparsity and the per-GPU column split.  The same
pattern applies in single-node multi-GPU and multi-node deployments.

\section{Experiments}

We evaluate the proposed PyTorch implementation of DuaLip along three axes:
(i) numerical parity with the production Scala implementation,
(ii) system-level performance and multi-GPU scaling, and
(iii) the impact of algorithmic enhancements such as preconditioning and regularization continuation.

We use synthetic matching data to enable controlled scaling of problem size and sparsity. The data generation procedure and complete experimental setup are described in Appendix~\ref{app:experiment_setting}.

\paragraph{Implementation parity} We first verify numerical equivalence between the PyTorch implementation and the original Scala solver. Figure~\ref{fig:scala_pytorch_parity} compares the dual objective trajectories across accelerated gradient descent (AGD) iterations in single- and multi-GPU settings.

The trajectories exhibit near-perfect overlap across all configurations, confirming implementation-level parity. To quantify the discrepancy, Figure~\ref{fig:scala_pytorch_relative_error} reports the relative error in dual objective compared to the Scala solver. In all experiments, the relative error falls below 1\% within the first 100 iterations and continues to decrease thereafter. These results demonstrate that the distributed PyTorch implementation faithfully reproduces the optimization dynamics of the production system.

\begin{figure}[ht!]
    \centering
    \includegraphics[width=0.75\linewidth]{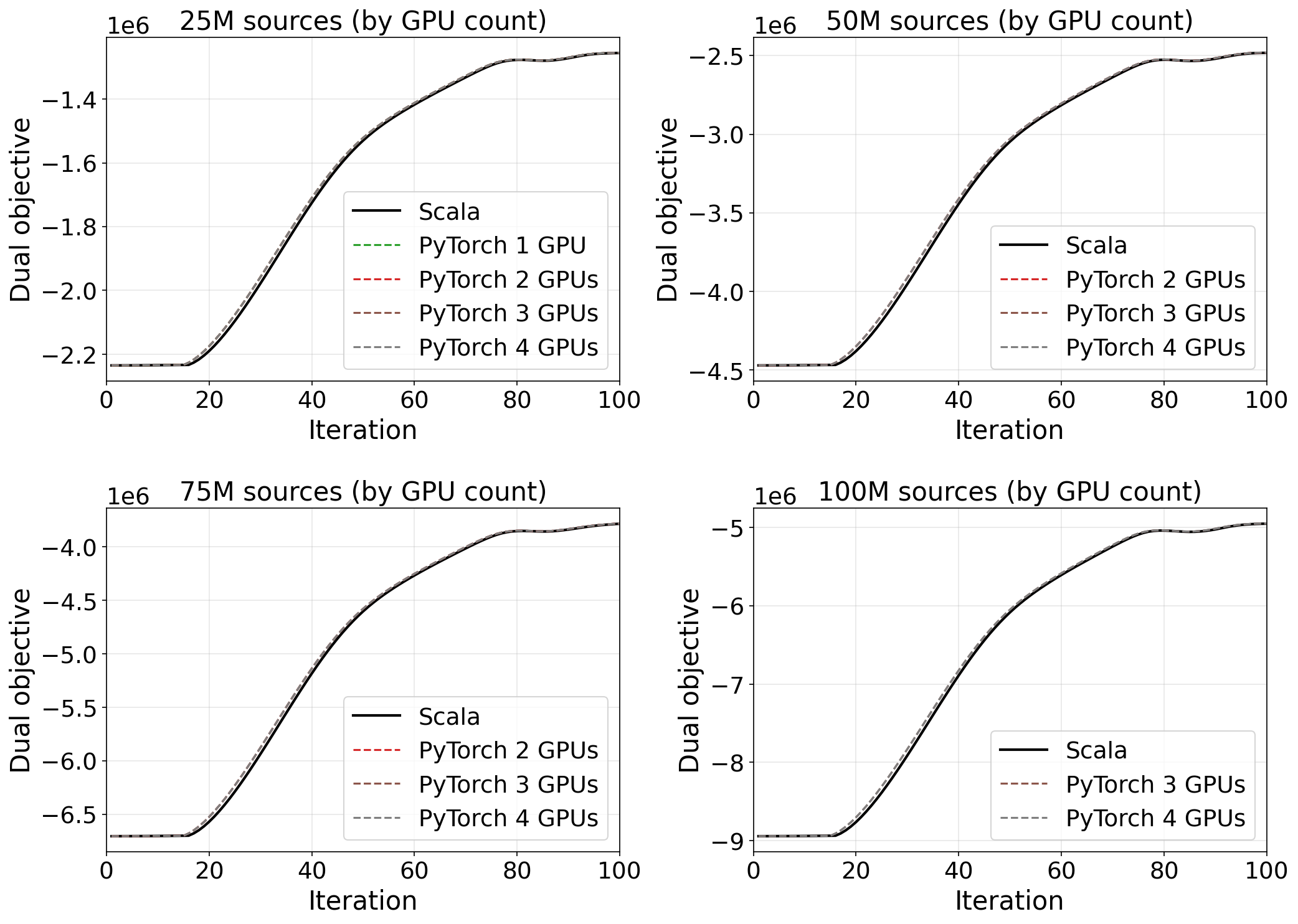}
    \caption{Scala–DuaLip (PyTorch) parity. Each panel shows the dual objective versus AGD iteration for the Scala and PyTorch implementations. The near-perfect overlap confirms numerical equivalence.}
    \label{fig:scala_pytorch_parity}
\end{figure}

\begin{figure}[ht!]
    \centering
    \includegraphics[width=0.75\linewidth]{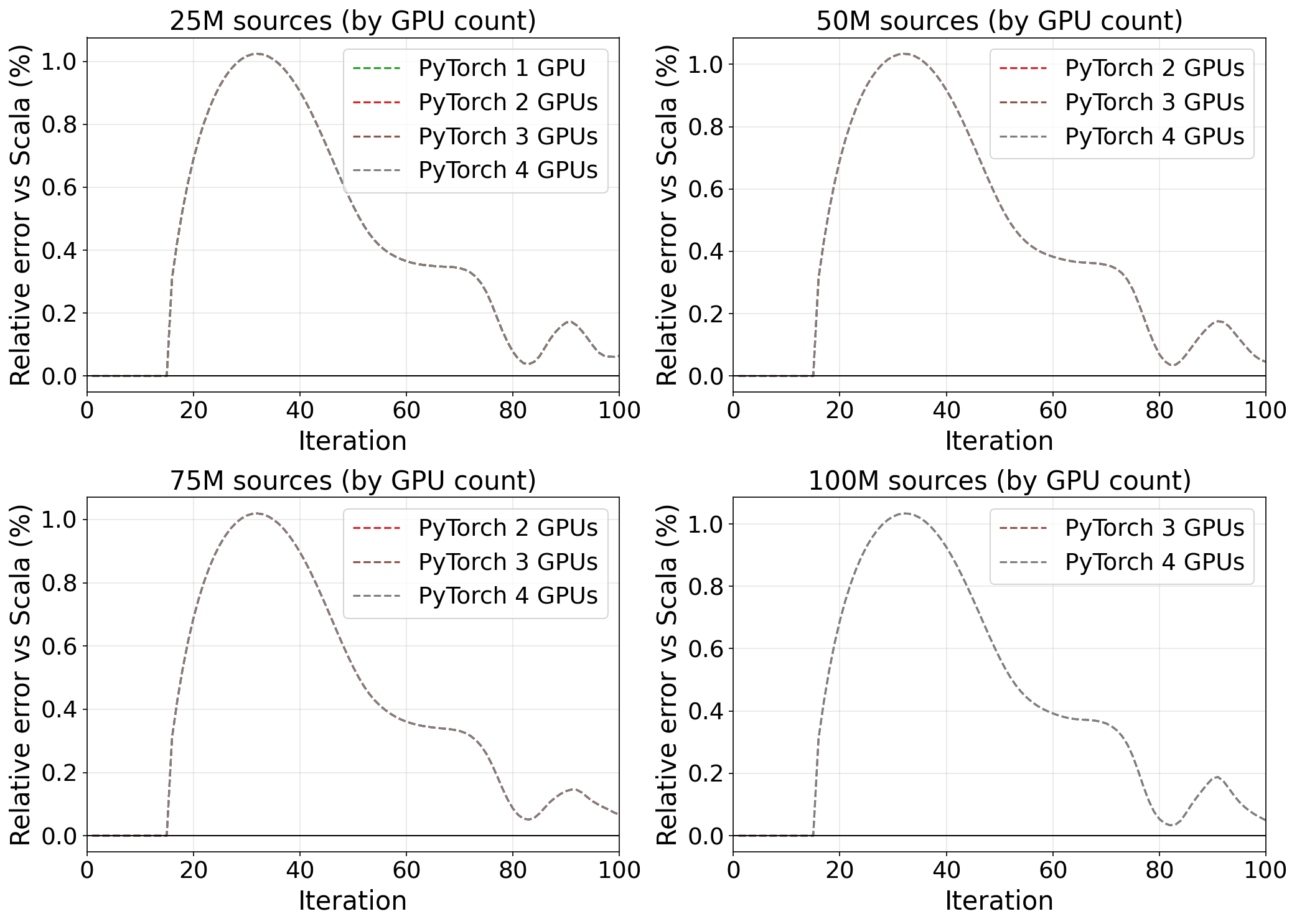}
    \caption{Relative error in dual objective compared to the Scala solver. The error drops below 1\% within the first 100 iterations across all settings.}
    \label{fig:scala_pytorch_relative_error}
\end{figure}

\paragraph{System Performance and Multi-GPU Scaling} We next compare per-iteration runtime between Scala (Spark-based) DuaLip and the PyTorch GPU implementation. Using a fixed random seed, we generate identical problem instances compatible with the Scala solver input schema and evaluate both systems under comparable configurations.

We do not include PDLP/cuPDLP in this runtime comparison, as existing GPU implementations operate on a single device and do not support the multi-GPU sharding regime required for our largest instances. Since our focus here is on scaling to extreme problem sizes under distributed execution, such a comparison would not reflect the intended deployment setting.

Table~\ref{tab:pytorch_scala_time_comparison} reports the average time per AGD iteration. For moderate problem sizes (25M sources), a single GPU already provides an order-of-magnitude improvement over Scala. As the problem size increases, model sharding across multiple GPUs becomes necessary to satisfy memory constraints. Beyond enabling larger instances, multi-GPU parallelization significantly reduces iteration time, achieving more than 10$\times$ speedup relative to Scala and near-linear scaling within the GPU setting.

\begin{table}[t]
\centering
\sisetup{
    table-number-alignment = center,
    round-mode = places,
    round-precision = 2
}
\begin{tabular}{l S *{4}{S}}
\toprule
&  & \multicolumn{4}{c}{\textbf{DuaLip (PyTorch)}} \\
\cmidrule(lr){3-6}
\textbf{Sources} & \textbf{Scala}
& {\textbf{1 GPU}} 
& {\textbf{2 GPUs}} 
& {\textbf{3 GPUs}} 
& {\textbf{4 GPUs}} \\
\midrule
25M  & 2.46 & 0.27 & 0.14 & 0.09 & 0.07 \\
50M  & 3.44 & {-}    & 0.27 & 0.18 & 0.13 \\
75M  & 2.63 & {-}    & 0.43 & 0.29 & 0.21 \\
100M & 3.33 & {-}    & {-}    & 0.37 & 0.27 \\
\bottomrule
\end{tabular}
\caption{Average time per AGD iteration (seconds). Number of destinations = 10{,}000, sparsity = 0.001. Multi-GPU sharding enables larger instances and yields substantial speedups over the Spark-based Scala implementation.}
\label{tab:pytorch_scala_time_comparison}
\end{table}

Figure~\ref{fig:scaling} further illustrates scaling behavior. Solve time decreases consistently as GPUs are added, and the achieved speedup approaches the ideal linear trend (3.86$\times$ on 4 GPUs versus the ideal 4$\times$). This demonstrates that the solver effectively parallelizes gradient computation with minimal communication overhead.

\begin{figure}[h]
    \centering
    \includegraphics[width=0.48\textwidth]{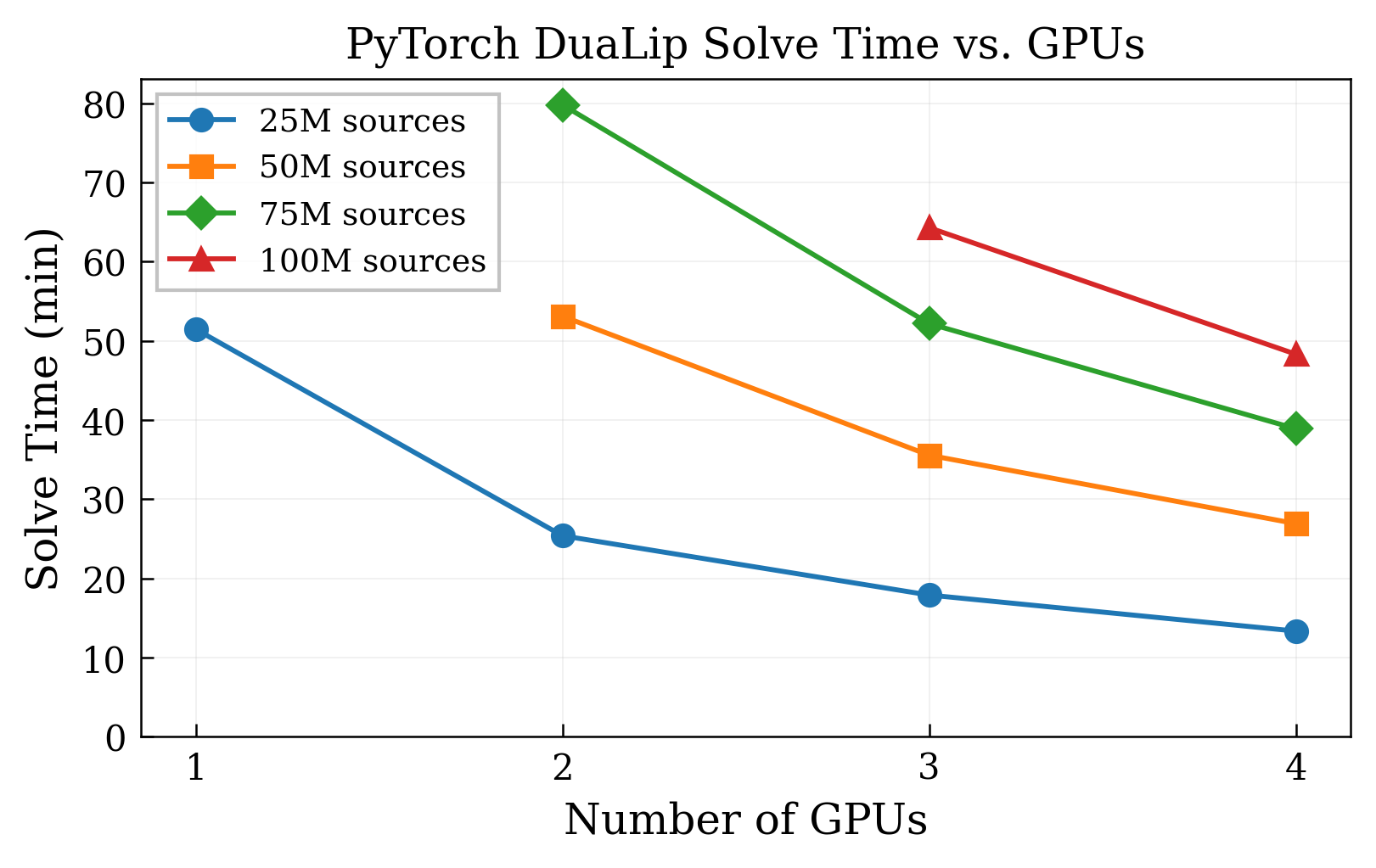}
    \hfill
    \includegraphics[width=0.48\textwidth]{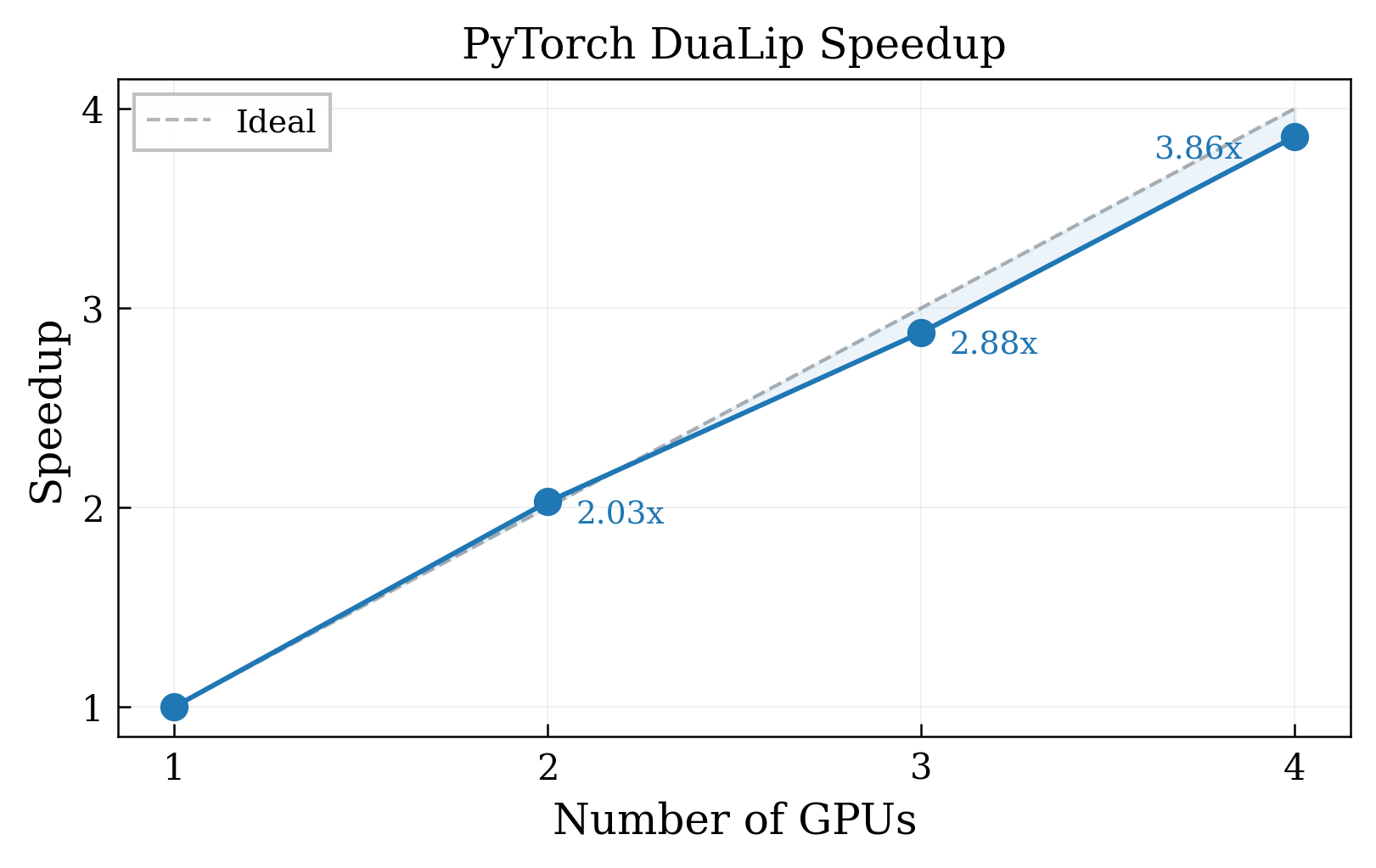}
    \caption{Scaling behavior across GPUs. (Left) Solve time versus number of GPUs for problem sizes between 25M and 100M sources. (Right) Speedup relative to a single GPU, approaching ideal linear scaling.}
    \label{fig:scaling}
\end{figure}

\paragraph{Effect of Algorithmic Enhancements}

We now isolate the impact of two algorithmic improvements introduced in Section~\ref{sxn:algorithm_enhancments}: diagonal preconditioning and regularization continuation.

\emph{Preconditioning.}
Figure~\ref{fig:preconditioning} plots $\log(|L - \hat{L}|)$, where $L$ is the dual objective and $\hat{L}$ is the converged reference value. With diagonal preconditioning, convergence accelerates substantially, particularly in early iterations. This confirms that scaling the dual variables mitigates ill-conditioning arising from heterogeneous constraint magnitudes.

\begin{figure}[ht!]
    \centering
    \includegraphics[width=10cm]{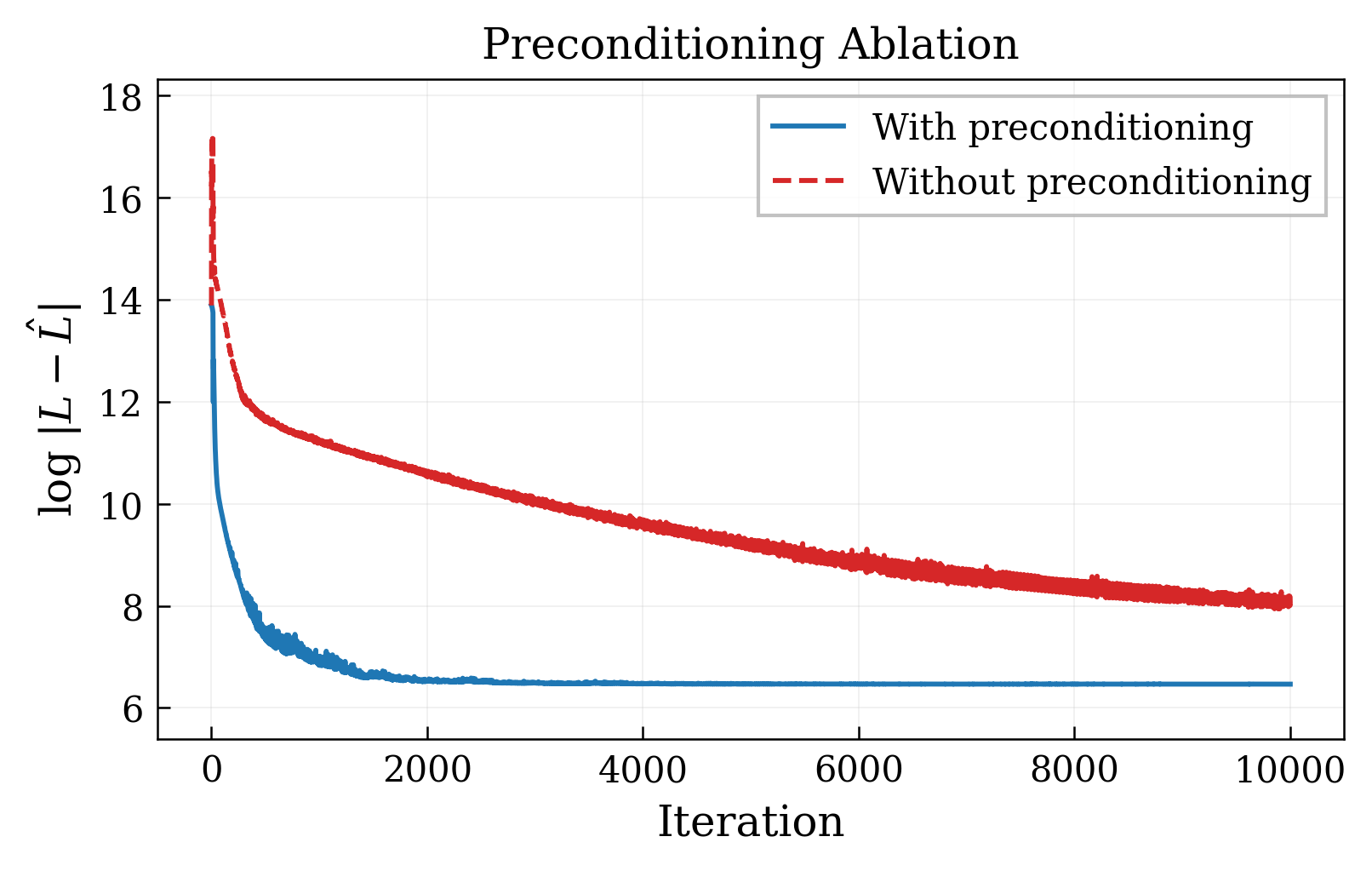}
    \caption{Effect of diagonal preconditioning. We report $\log(|L - \hat{L}|)$ for a 25M-source instance (10k destinations, 0.1\% sparsity). Preconditioning significantly improves early-stage convergence.}
    \label{fig:preconditioning}
\end{figure}

\emph{Regularization Continuation.}
Figure~\ref{fig:gamma_decay} evaluates the continuation strategy for the ridge parameter $\gamma$. Starting from a larger $\gamma$ stabilizes and accelerates early optimization, while gradual decay ensures the final solution closely approximates the unregularized LP optimum. Decaying $\gamma$ from 0.16 to 0.01 (halved every 25 iterations) yields faster convergence compared to using a fixed regularization level.

\begin{figure}[ht!]
    \centering
    \includegraphics[width=10cm]{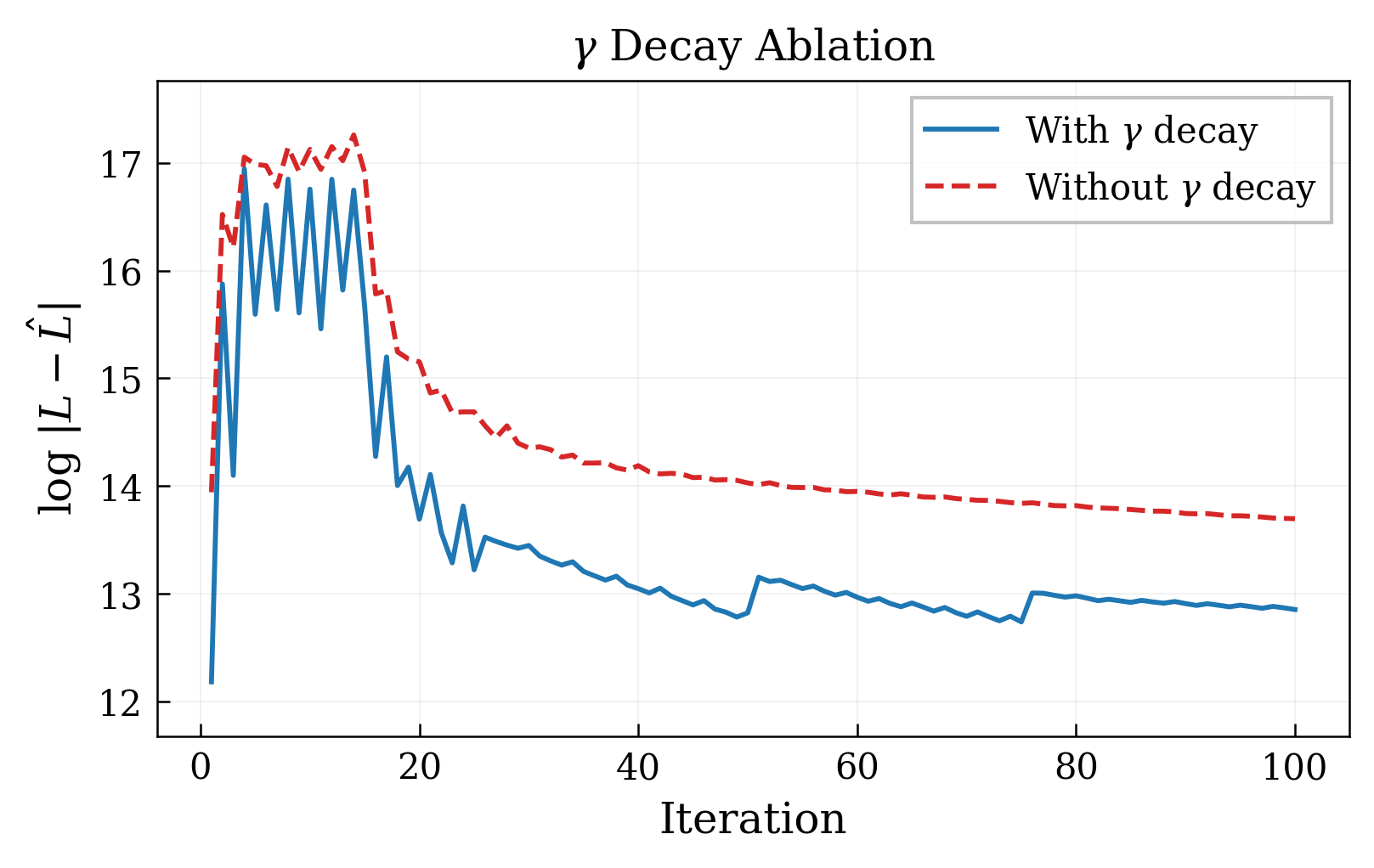}
    \caption{Effect of regularization continuation. Decaying $\gamma$ during optimization accelerates convergence while preserving solution fidelity.}
    \label{fig:gamma_decay}
\end{figure}

\paragraph{Reproducibility} Code for reproducing all synthetic benchmarking experiments is available at \url{https://github.com/linkedin/DuaLip}.

\section{Discussion and Limitations}

This report provided an overview of the principal architectural choices, algorithmic improvements, and GPU implementation strategies needed to address extreme-scale matching problems in an industrial setting. A natural next step is to further evaluate the core dual-ascent method on standard linear programming benchmarks beyond matching. In addition, because suitable real-world public datasets for extreme-scale matching LPs are scarce, the current evaluation relies on large synthetic instances. Establishing standardized benchmark suites and data sources would enable more comprehensive and reproducible empirical studies of extreme-scale matching LPs.

%% file: appendix.tex
\section{Additional theory}

\subsection{Proof of Lemma \ref{lemma:precon}}

\begin{proof}




Then
\[
\Ab\Ab^\top
=\begin{bmatrix}\Ab_1&\cdots&\Ab_I\end{bmatrix}
 \begin{bmatrix}\Ab_1^\top\\ \vdots\\ \Ab_I^\top\end{bmatrix}
=\sum_{i=1}^I \Ab_i\Ab_i^\top.
\]
Taking expectations and using i.i.d.\ blocks gives
$\EE[\Ab\Ab^\top]=\sum_{i=1}^I \EE[\Ab_i\Ab_i^\top]=I\,\Mb$
with $\Mb=\EE[\Ab_1\Ab_1^\top]$.

For the normalized matrix,
\[
\EE[\widetilde{\Ab}\widetilde{\Ab}^\top]
=\Db_{\mathrm{exp}}\,\EE[\Ab\Ab^\top]\,\Db_{\mathrm{exp}}
= I\,\Db_{\mathrm{exp}}\,\Mb\,\Db_{\mathrm{exp}}.
\]
By definition of $\Db_{\mathrm{exp}}$,
the $(r,r)$ entry equals
\(
I\,\EE\|\Ab_{r*}\|_2^2\,/\,\EE\|\Ab_{r*}\|_2^2=1
\),
so $\mathrm{diag}\!\big(\EE[\widetilde{\Ab}\widetilde{\Ab}^\top]\big)=\Ib$.

For $r\neq s$,
\[
\big|\EE[\widetilde{\Ab}\widetilde{\Ab}^\top]_{rs}\big|
=\frac{\big|\EE\langle \Ab_{r*},\Ab_{s*}\rangle\big|}
{\sqrt{\EE\|\Ab_{r*}\|_2^2\,\EE\|\Ab_{s*}\|_2^2}}
\;\le\;\eta.
\]
Thus each row of $\EE[\widetilde{\Ab}\widetilde{\Ab}^\top]$ has diagonal entry $1$
and the sum of absolute off-diagonals at most $(m-1)\eta$.
By Gershgorin’s disc theorem,
all eigenvalues lie in $[\,1-(m-1)\eta,\;1+(m-1)\eta\,]$, giving
\(
\kappa\le\frac{1+(m-1)\eta}{1-(m-1)\eta}.
\)
\end{proof}

\subsection{Primal feasibility convergence}

We extend the theoretical analysis of \cite{basu2020eclipse} and \cite{ramanath2021efficient}. In particular, we bound the $\ell_2$-norm primal infeasibility as a function of dual objective suboptimality, which allows one to apply standard first order convergence theory to obtain primal feasibility guarantees from the dual ascent method presented in Section \ref{sxn:prelim}. 

\vspace{0.5cm}
\begin{lemma}
    Let $\lambdab^*$ be the unique maximizer of the concave dual $g$ associated with the LP \ref{eqn:LP}. If $\lambdab \ge \zero$, then, primal infeasibility is bounded above as
    \begin{equation}
        \label{eq:dual-gap-to-residual}
        \bigl\|\left(\Ab\xb_\gamma^*(\lambdab)-\bb\right)_{+}\bigr\|_2
        \leq
        \sqrt{\,2L\bigl(g(\lambdab^*)-g(\lambdab)\bigr)}\,,
    \end{equation}
    where $L=\|\Ab\|_2^2/\gamma$ is the Lipschitz constant of $\nabla g$, and for a vector $\vb$, the operator $(\vb)_+$ denotes the componentwise maximum with zero, i.e.,
    \[
    (\vb_+)_i = \max(v_i, 0).
    \]
\end{lemma}


\begin{proof}
Let $g$ be concave with $L$-Lipschitz gradient and define
\[
\lambdab_+ \;=\; \Pi_{\mathbb{R}_+^m}\!\left(\lambdab + \frac{1}{L}\nabla \gb(\lambdab)\right),
\qquad
\db := \lambdab_+ - \lambdab,
\qquad
\nabla_+ \gb(\lambdab):=\max(\nabla \gb(\lambdab), \zero).
\]
Note that $\nabla_+ \gb(\lambdab) = (\Ab \xb^\star_{\gamma}(\lambdab) - \bb)_{+}$. By $L$-smoothness of $g$ and concavity,
\begin{equation}
\label{eq:ascent-lemma}
g(\lambdab_+) - g(\lambdab)
\;\ge\;
\langle \nabla \gb(\lambdab), \db \rangle
-
\frac{L}{2}\|\db\|_2^2.
\end{equation}

Let $\yb := \lambdab + \frac{1}{L}\nabla \gb(\lambdab)$. Since $\lambdab_+=\Pi_{\mathbb{R}_+^m}(\yb)$,
the projection optimality condition gives
\[
\langle \yb-\lambdab_+, \ub-\lambdab_+\rangle \le 0
\qquad \forall \ub \in \mathbb{R}_+^m.
\]
Choosing $\ub=\lambdab$ (valid since $\lambdab\ge 0$) yields
\[
\left\langle \lambdab + \frac{1}{L}\nabla \gb(\lambdab) - \lambdab_+,\, \lambdab - \lambdab_+ \right\rangle \le 0.
\]
Substituting $\db=\lambdab_+-\lambdab$ and simplifying,
\[
\left\langle \frac{1}{L}\nabla \gb(\lambdab) - \db,\; -\db \right\rangle \le 0
\;\;\Longrightarrow\;\;
\left\langle \nabla \gb(\lambdab),\db\right\rangle \ge L\|\db\|_2^2.
\]

Plugging this into~\eqref{eq:ascent-lemma} gives
\[
g(\lambdab_+) - g(\lambdab)
\;\ge\;
L\|\db\|_2^2 - \frac{L}{2}\|\db\|_2^2
=
\frac{L}{2}\|\db\|_2^2.
\]

Next, note that
\[
L\db = L(\lambdab_+-\lambdab)
= L\left(\Pi_{\mathbb{R}_+^m}\!\left(\lambdab + \frac{1}{L}\nabla \gb(\lambdab)\right)-\lambdab\right)
= \max\!\bigl(\nabla \gb(\lambdab),\, -L\lambdab\bigr),
\]
so $\|L\db\|_2 \ge \|\nabla_+ \gb(\lambdab)\|_2$.
Therefore
\[
g(\lambdab_+) - g(\lambdab)
\;\ge\;
\frac{L}{2}\|\db\|_2^2
=
\frac{1}{2L}\|L\db\|_2^2
\;\ge\;
\frac{1}{2L}\|\nabla_+ \gb(\lambdab)\|_2^2.
\]
Finally, 
\[
g(\lambdab^\star) - g(\lambdab) \;\ge\; g(\lambdab_+) - g(\lambdab)
\;\ge\;
\frac{1}{2L}\|L\db\|_2^2
\;\ge\; \frac{1}{2L} \|(\Ab \xb^\star(\lambdab) - \bb)_+\|^2_2.
\]
\end{proof}

\section{Experiment setting}\label{app:experiment_setting}

\paragraph{Synthetic LP construction.}
We construct synthetic instances by first generating a sparse bipartite interaction graph and then assigning values and constraint coefficients on its edges. Given a target number of ``requests'' \(I\), ``resources'' \(J\), and a target sparsity level, we draw a lognormal ``breadth'' parameter for each resource \(j\), normalize these to obtain probabilities \(p_j\), and sample the number of incident requests \(K_j \sim \mathrm{Poisson}(p_j I \nu)\), truncated at \(I\), where $\nu$ is the desired average number of nonzeros per row. For each resource \(j\), we then select \(K_j\) distinct requests and create edges \((i,j)\). On each edge, we draw a resource-specific value scale \(v_j\), a request-specific responsiveness factor \(u_i\), and multiplicative noise \(\varepsilon_{ij}\), and define a nonnegative value coefficient
\[
c_{ij} = \min\bigl(v_j u_i \varepsilon_{ij},\, c_{\max}\bigr).
\]
Constraint coefficients are taken to be scaled versions of these values,
\[
a_{ij} = s_j c_{ij},
\]
where the per-resource scale \(s_j\) is also drawn from a lognormal distribution. This construction yields a sparse matrix \(\Ab\) whose rows differ both in support size and magnitude (often by several orders), and a matching value matrix \(\Cb\) with the same sparsity pattern. 

\paragraph{Source capacities and right-hand side.}
Right-hand side source capacities \(b_j\) are chosen to make a nontrivial fraction of constraints active. Instead of taking \(b_j\) proportional to the sum \(\sum_i a_{ij}\), we approximate the maximum feasible load each resource could receive under the per-request simplex constraint (each request can allocate to at most one resource) by a greedy assignment: for each request \(i\), we identify the incident edge with the largest \(a_{ij}\) and assign that amount to the corresponding resource. Summing these contributions over requests gives a ``greedy load'' \(\ell_j\) for each resource, and we set 
\[
b_j = \rho_j \bigl(\ell_j + \varepsilon\bigr),
\]
where \(\rho_j\) is drawn uniformly from \([0.5, 1.0]\) and \(\varepsilon > 0\) is a small constant. This ensures that some resource constraints are binding while others remain slack in the optimal solution. The final LP data passed to the solver is the sparse CSC representation of $\Ab$, the corresponding value matrix (with signs adjusted to match our minimization convention), and the source capacity vector $\bb$.




\paragraph{Optimization algorithm.}
We optimize the regularized dual objective using a variant of Nesterov accelerated gradient ascent, following the implementation in \cite{dualipcode} (\texttt{AcceleratedGradientDescent.scala}), translated directly into PyTorch.

The method maintains a running estimate of the local Lipschitz constant based on recent gradient information, which is used to adaptively select the step size at each iteration. This allows the algorithm to respond to changes in local curvature while preserving the accelerated structure.

To ensure numerical stability, the step size is capped by a maximum allowable value. If this cap is set too aggressively, the method may diverge due to curvature underestimation; if set too conservatively, convergence slows substantially. In practice, this cap plays a critical role in balancing robustness and convergence speed.

All experiments use fixed hyper-parameters:
\texttt{max-step-size = 1e-3} and 
\texttt{initial-step-size = 1e-5}, 
and ridge parameter $\gamma = 0.01$ unless otherwise noted.
These values were held constant across problem sizes to ensure comparability.
Optimization is terminated after a fixed number of iterations. 



